\documentclass[11pt]{article}

\usepackage[margin=1in, paperwidth=8.5in, paperheight=11in]{geometry}
\usepackage{graphicx, mathtools}
\usepackage{caption}
\usepackage{subcaption}
\usepackage{enumerate}
\usepackage{amsmath,amsthm,amssymb}
\usepackage{setspace}
\usepackage[parfill]{parskip}
\usepackage{comment}
\usepackage{apacite}

\graphicspath{ {images/} } 

\theoremstyle{definition}

\newtheorem*{definition*}{Definition}
\newtheorem{remark}{Remark}
\newtheorem*{remark*}{Remark}

\theoremstyle{plain}
\newtheorem{proposition}{Proposition}

\newtheorem{lemma}{Lemma}

\theoremstyle{remark}

\newcommand{\cP}{\mathcal P}

\newcommand{\beq}{\begin{equation}}
\newcommand{\eeq}{\end{equation}}
\newcommand{\beqa}{\begin{eqnarray*}}
\newcommand{\eeqa}{\end{eqnarray*}}

\begin{document}

\title{A Simple and Efficient Method to Compute a Single Linkage Dendrogram}
\author{
	Huanbiao Zhu \\
	Department of Statistics \\
	University of Washington \\
	\texttt{huanbz@uw.edu}
	\and
	Werner Stuetzle \\
	Department of Statistics \\
	University of Washington \\
	\texttt{wxs@uw.edu}
}
\date{Oct 31, 2019}
\maketitle
\doublespacing

\begin{abstract}
\noindent We address the problem of computing a single linkage dendrogram. A possible approach is to: (i) Form an edge weighted graph $G$ over the data, with edge weights reflecting dissimilarities. (ii) Calculate the MST $T$ of $G$. (iii) Break the longest edge of $T$ thereby splitting it into subtrees $T_L$, $T_R$. (iv) Apply the splitting process recursively to the subtrees. This approach has the attractive feature that Prim's algorithm for MST construction calculates distances as needed, and hence there is no need to ever store the inter-point distance matrix. The recursive partitioning algorithm requires us to determine the vertices (and edges) of $T_L$ and $T_R$. We show how this can be done easily and efficiently using information generated by Prim's algorithm without any additional computational cost.
\end{abstract}

\section{Introduction}
In a generic clustering problem we are given a collection $V$ of $n$ objects and a function $d(v_i, v_j)$ measuring the dissimilarity between objects $v_i$ and $v_j$. The goal is to partition $V$ into subsets(“clusters”) such that observations in the same cluster are similar and dissimilar from observations in other clusters.


\subsection{Hierarchical Clustering}
Clustering methods come in two varieties, $\it flat$ and $\it hierarchical$. Flat methods require the user to
provide a target number $k$ of clusters and will then generate a
partition $\cP_k$ = $\{C_1,\ldots C_k\}$ of $V$. Hierarchical methods differ from flat methods in that they generate a hierarchy of partitions ${\cal P}_1,...,{\cal P}_n$. A sequence of partitions is called hierarchical if each cluster in $\cP_i$ is the union of clusters in $\cP_{i + 1}$.

Hierarchical methods can be agglomerative or divisive. Hierarchical agglomerative clustering (HAC) methods generate partitions $\cP_1,\ldots, \cP_n$ by iterative merging. Initially, every object forms a cluster. Then we repeatedly merge the two clusters with the minimum distance (dissimilarity) until only one cluster is left. 
Different hierarchical methods differ in the definition of the distance between clusters. We will focus on single linkage clustering where $D(C_1, C_2)$ is defined as the minimum distance between an object in $C_1$ and an object in $C_2$, i.e. $D(C_1, C_2)=min(d(v_i,v_j))$ with $v_i\in C_1$ and $v_j\in C_2$. 




\begin{figure}[!htb]
	\captionsetup{font=small}
	\hfill
	\begin{subfigure}[h]{0.49\linewidth}
		\includegraphics[width=0.9\linewidth]{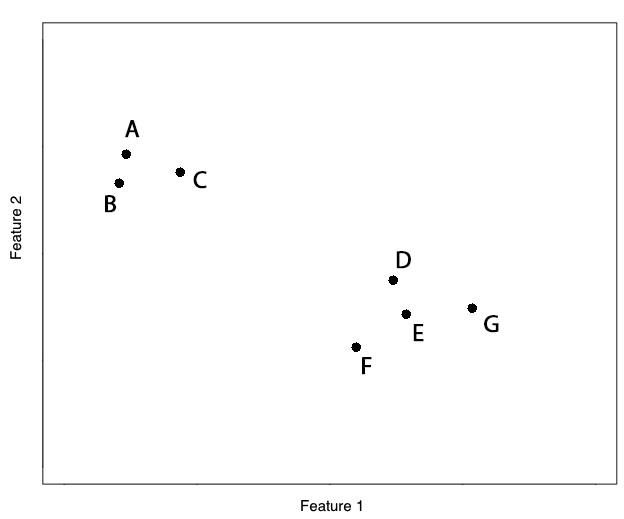}
		\caption{}
	\end{subfigure}
	\hfill
	\begin{subfigure}[h]{0.49\linewidth}
		\includegraphics[width=0.8\linewidth]{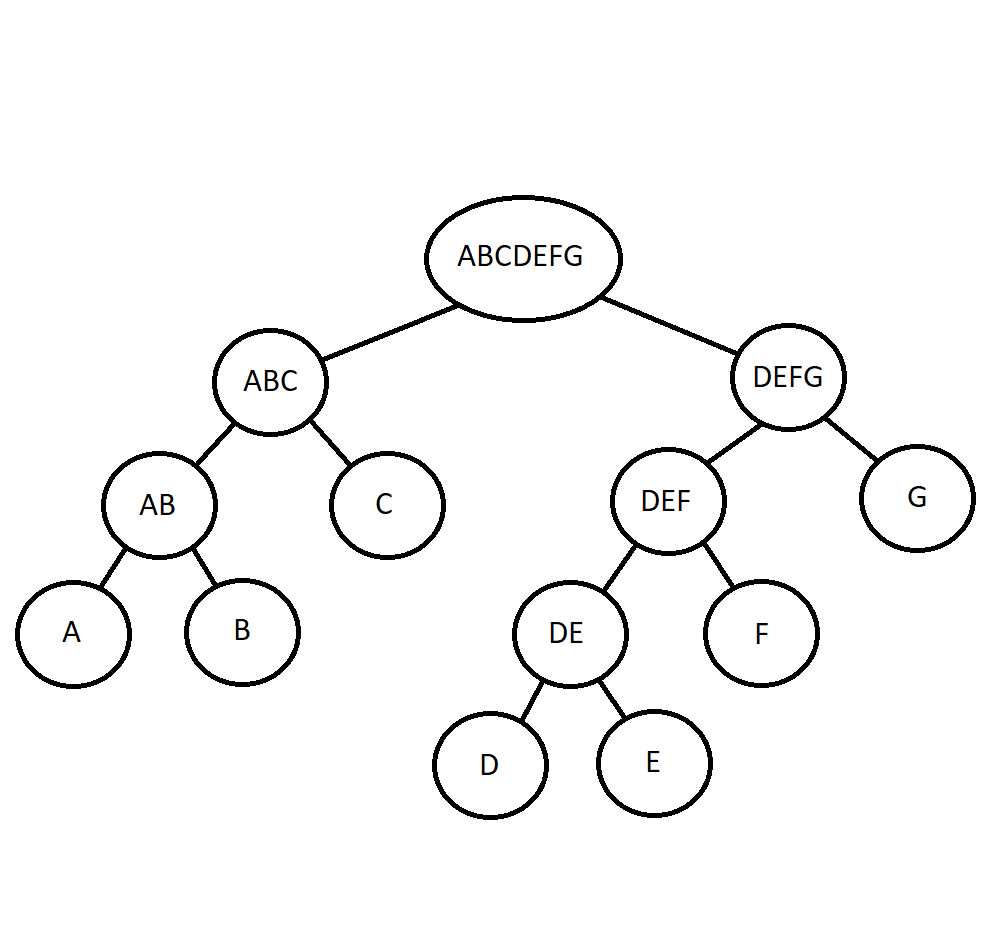}
		\caption{}
	\end{subfigure}
	\hfill
	\begin{subfigure}[h]{0.49\linewidth}
		\includegraphics[width=0.9\linewidth]{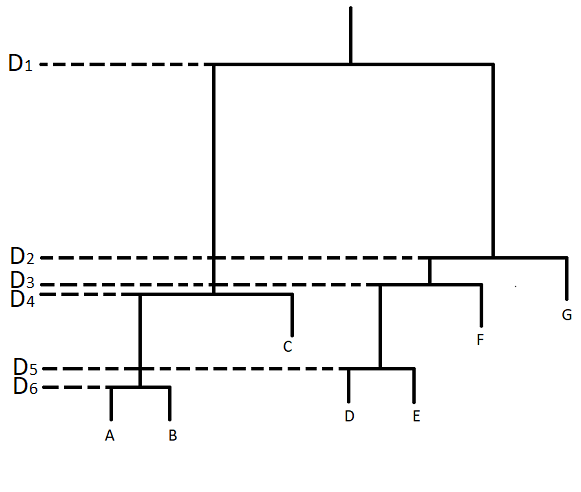}
		\caption{}
	\end{subfigure}
	\hfill
	\begin{subfigure}[h]{0.49\linewidth}
		\includegraphics[width=0.9\linewidth]{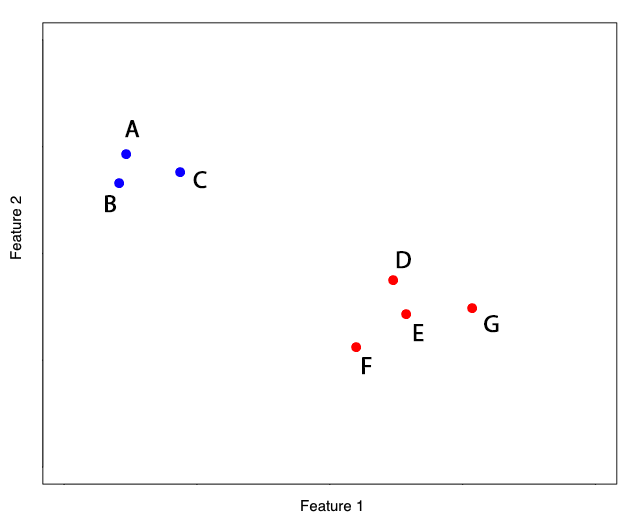}
		\caption{}
	\end{subfigure}%
	\caption{Single linkage clustering. \textbf{(a)} Sample data set.  \textbf{(b)} Binary tree of clusters. \textbf{(c)} Dendrogram of data set with merge distances. \textbf{(d)} Partition of data set with distance threshold $D_1$}
\end{figure}

\subsection{Dendrogram}
The result of the merge process can be represented as a binary tree
with $n$ leaves. Each node of the tree represents a subset of the observations, called the {\it node members}. Each leave represents an individual observation. Each internal node represents the union of the members of its daughter nodes and is associated with a {\it merge  distance}, the distance between the two clusters being merged.

A layout of this tree where the root is at the top, the leaves are at the bottom, and the vertical coordinate of an interior node is the merge distance, is called a dendrogram. 

Figure 1 shows a dataset (a), the binary tree generated by single linkage clustering (b), the corresponding single linkage dendogram (c), and the partition of the data into two clusters (d).

\subsection{Extracting Partitions}
Any subtree of a dendrogram defines a partition of $V$; the members of the leaves are the clusters. The most commonly used pruning method is dendrogram cutting: choose a distance threshold $D^*$ and eliminate all nodes with merge distance less than $D^*$. Figure 1(d) shows the partition of our sample data set obtained by dendrogram cutting with distance threshold $D_1$. There are alternative pruning methods; see for example \cite{runtpruning}.


\subsection{Computing the Single Linkage Dendrogram}



The single linkage dendrogram could in principle be computed using the iterative merging algorithm sketched in Section 1.1. In practice, however this is not an attractive option because it requires storing the interpoint distance matrix. An alternative was suggested by Gower and Ross \cite{mst}. They proposed to first compute the minimal spanning tree (MST) T of the data and then obtain the single linkage dendrogram by recursive partitioning: break the longest edge of the T, thereby splitting T into two subtrees $T_L$ and $T_R$, and then apply the splitting operation recursively to the two subtrees. The key advantage of this approach is that the MST can be computed using Prim's algorithm \cite{prim} without ever storing the interpoint distance matrix. Prim's algorithm produces a list of MST edges. The remaining problem is to determine the edges of $T_L$ and $T_R$, and thereby the node members of the corresponding dendrogram nodes. We propose a simple and efficient method of identifying the vertices (and edges) of $T_L$ and $T_R$ using information generated by Prim's algorithm.





\section{Prim's Algorithm and Prim's Order}


\subsection{Prim's algorithm}
Prim's algorithm finds a minimal spanning tree of a weighted connected graph $G=(V,E,W)$ (In the application of the MST to single linkage clustering, $G$ is typically the complete graph over some set $V$ of points in the Euclidean space and the edge weights are the Euclidean distances). The algorithm starts a tree fragment by choosing an arbitrary seed vertex $v_{seed}$ and then progressively connects the out-vertices(vertices that have yet been connected) to the fragment. Below is an outline of the algorithm.
\begin{enumerate}[Step 1:]
	\item Initialization: Choose an arbitrary vertex $v_{seed}\in V$ and set $V_{mst}=\{v_{seed}\}, E_{mst}=\emptyset$.
	\item Iteration: While $V_{mst}\neq V:$
	\begin{enumerate}
		\item Find an edge $(u, v)$ such that $u\in V_{mst}, v \notin V_{mst}$ and $W(u, v)$ is minimized.
		\item Add $v$ into $V_{mst}$, $(u, v)$ into $E_{mst}$
	\end{enumerate}
\end{enumerate}

The MST is unique if all edges have distinct weights.

\subsection{Prim's order}
When we apply Prim's algorithm to $G$, each iteration adds one vertex to the fragment. This defines an order for the vertices which we call \textit{Prim's order (of G)}. Let $P(v_k)$ denote the position of $v_k$ in Prim's order. The order depends on the choice of the seed vertex $v_{seed}$. The seed vertex is arbitrary unless otherwise noted. By default, we define $P(v_{seed})=1$. Prim's order of vertices also induces an order of the MST edges: For an MST edge $e=(v_i,v_j)$, define $P(e)=max(P(v_i),P(v_j))$.

\begin{remark}
If there are no tied edge weights in $G$, then the MST $T$ and Prim's order of $G$ are unique. If there are tied edge weights, then there might be more than one MST as well as more than one Prim's order for a given seed vertex.
\end{remark}

\section{Finding Connected Components after Breaking the Longest MST Edge}
Based on Prim's algorithm and Prim's order, we will introduce a method which can efficiently find two connected components obtained by breaking the longest MST edge. For sake of simplicity, let us assume for the moment that the graph $G$ has no tied edge weights. We will treat the general case in Section 4.  

\begin{figure}[!htb]
     \centering
     \captionsetup{justification=justified, font=small}
    \begin{subfigure}[t]{0.49\textwidth}
           \centering
        \raisebox{-\height}{\includegraphics[width=0.9\textwidth]{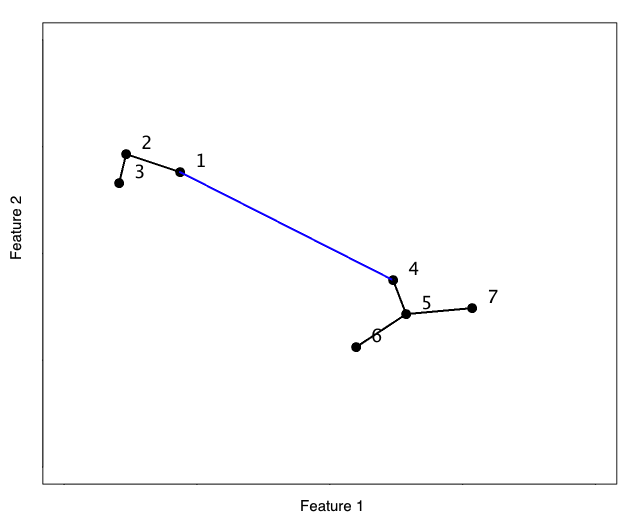}}
        \caption{}
    \end{subfigure}
    \hfill
    \begin{subfigure}[t]{0.49\textwidth}
        \centering
        \raisebox{-\height}{\includegraphics[width=0.9\textwidth]{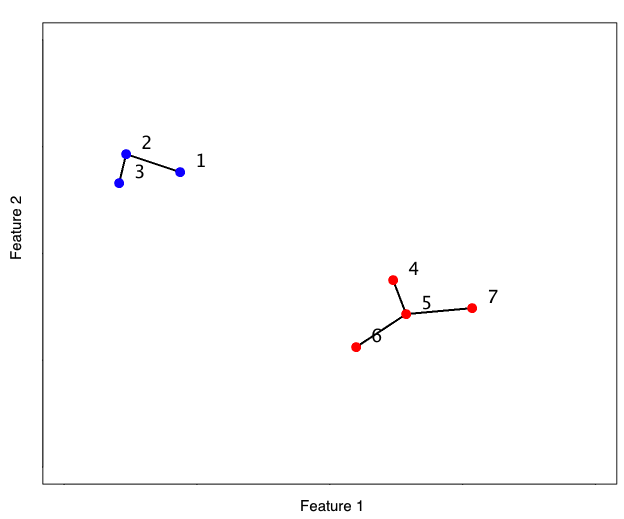}}
        \caption{}
    \end{subfigure}
    \begin{subfigure}[t]{0.49\textwidth}
        \centering
        \raisebox{-\height}{\includegraphics[width=0.9\textwidth]{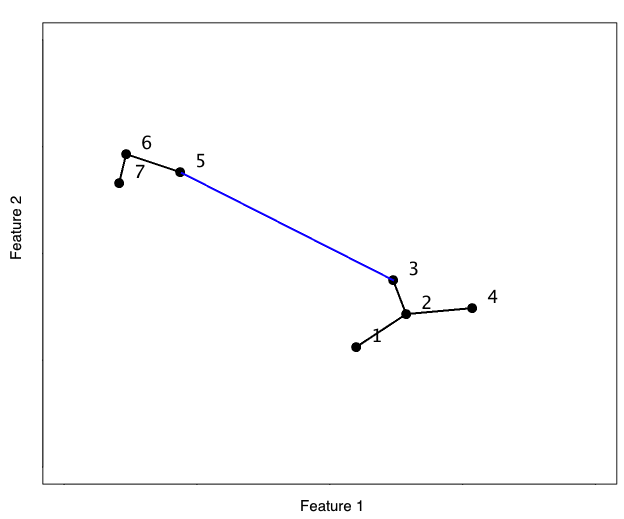}}
        \caption{}
    \end{subfigure}
    \hfill
    \begin{subfigure}[t]{0.49\textwidth}
        \centering
        \raisebox{-\height}{\includegraphics[width=0.9\textwidth]{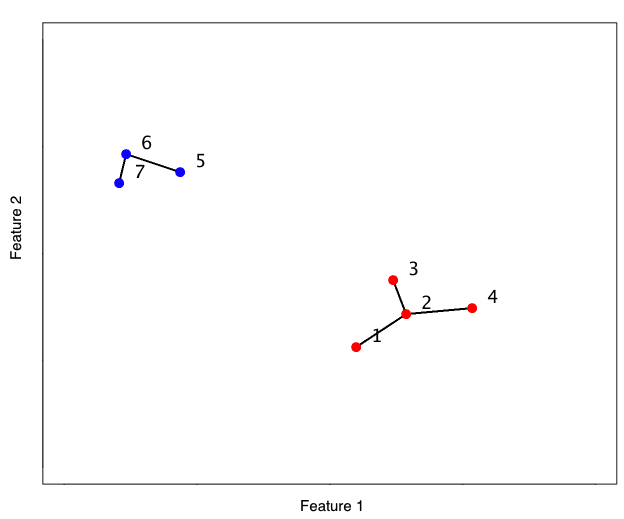}}
    \caption{} 
    \end{subfigure}
    \caption{Same MST with different Prim's orders and the subtrees after breaking the longest edge. \textbf{(a)} MST with the longest edge marked in blue; numbers next to the vertices indicate the Prim's orders. \textbf{(b)} Subtrees after breaking the longest edge in (a). \textbf{(c)} Same MST as in (a); different Prim's order. \textbf{(d)} Subtrees after breaking the longest edge in (c).}
\end{figure}

Take a look at Figure 2. Panels (a) and (c) show the sample data from Figure 1 and the MST. The numbers next to the vertices indicate their Prim's orders. The orders in Panels $(a)$ and $(c)$ are different because the seed vertices are different. The longest edge in (a) (marked in blue) connects vertices with Prim's order 1 and 4 while the same longest edge in (c) connects vertices with Prim's order 5 and 3.



When we break the longest edge, we obtain two subtrees which are shown in panels (b) and (d). We notice that the vertices in the two subtrees share a common characteristic. In (b), the vertices in one subtree all have Prim's order less than 4 and the vertices in the other subtree all have Prim's order greater than or equal to 4. Similarly, in (d), the vertices in one subtree all have Prim's order less than 5 and the vertices in the other subtree all have Prim's order greater than or equal to 5. This suggests that the vertex sets and edge sets of the two subtrees can simply be determined based on their Prim's orders. Notice that 4 and 5 are the Prim's orders of the longest edge in both cases. This suggests the following proposition.

\begin{proposition}
Let $G$ be a connected edge weighted graph with distinct edge weights. Applying Prim's algorithm to $G$ will result in a unique MST $T$ and a unique Prim's order $P$ (for some arbitrary seed vertex). Let $e_{max}=(v_i,v_j)$ be the longest MST edge. Breaking $e_{max}$ splits $T$ into two subtrees $T_L$ and $T_R$ with vertex sets $V_L,V_R$ and edge sets $E_L,E_R$. Then 
	\[
	V_L = \{v: P(v) < P(e_{max})\} \text{ and } V_R = \{v: P(v) \geq P(e_{max})\}
	\]
	\[
	E_L = \{e: P(e) < P(e_{max})\} \text{ and } E_R = \{e: P(e) > P(e_{max})\}
	\]
	where $P(e_{max})=max(P(v_i),P(v_j))$.
\end{proposition}

To prove the proposition, we use the following lemma.

\begin{lemma}
	Let $G$ be a connected edge weighted graph with distinct edge weights. Let $T$ be the minimal spanning tree of $G$. If Prim's algorithm is applied to $G$ and $T$ with the same seed vertex, then the Prim's orders of $G$ and $T$ are the same.
\end{lemma}
\begin{proof}
Let $E_G$ be the set of edges in $G$ and $E_T$ be the set of edges in $T$. Applying Prim's algorithm to $G$ will result in the MST $T$ and the unique Prim's order $P$. Let $e$ be an edge in $E_G-E_T$. Since $G$ has distinct edge weights, by removing $e$ from $G$, $T$ will still be the MST of $G-e$ and $P$ will still be the Prim's order of $G-e$. Repeatedly remove edges from $G$ until $G=T$. This shows that $G$ and $T$ have the same Prim's order.
\end{proof}

\begin{proof}[Proof of Proposition 1]
As shown in Lemma 1, it is sufficient to prove Proposition 1 for the MST $T$ rather than the original graph $G$. First, apply Prim's algorithm to $T$. Without loss of generality, supppose the seed vertex is in $T_L$. Since all edges in $T_L$ are shorter than $e_{max}$ and the only edge between $T_L$ and $T_R$ is $e_{max}$, this implies that all edges in $T_L$ must be joined to the fragment before $e_{max}$. Therefore, all remaining edges in $T_R$ must be joined after $e_{max}$. This also implies that all vertices in $V_L$ must be connected before any vertices in $V_R$.
\end{proof}

\begin{figure}[!htb]
	\centering
	\includegraphics[width=0.6\linewidth]{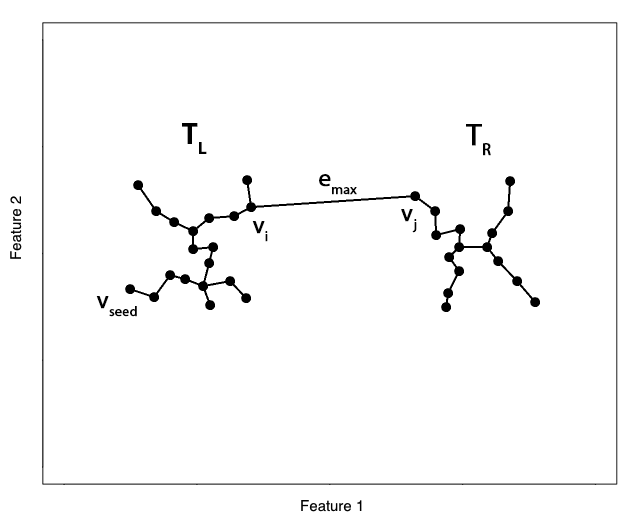}
	\caption{Illustrating graph for Proposition 1}
\end{figure}

While we have proved that Proposition 1 holds for breaking the longest edge once, it remains to be shown that the method for finding $T_L$ and $T_R$ stated in Proposition 1 can be applied recursively to $T_L$ and $T_R$. Note that $T_L$ and $T_R$ are themselves MSTs for their respective vertex sets. Let $P$ denote Prim's order of $T$ with seed vertex $v_{seed}$. Let $P_L$ be Prim's order of $T_L$ with seed vertex $v_{seed}$ and let $P_R$ be Prim's order of $T_R$ with seed vertex $v_j$. Then 
\[
P_L(v)=P(v) \hspace{0.3cm} \forall v\in T_L \hspace{0.2cm} \text{ and } \hspace{0.2cm} P_R(v)=P(v)-P(v_j)+1 \hspace{0.3cm} \forall v\in T_R.
\]
In other words, on their respective subtrees, $P_L$ and $P_R$ are "equivalent" to $P$.
This implies that Proposition 1 can be applied recursively until every vertex is an isolated vertex. Hence, we only need to compute the MST for $G$ and then use its Prim's order to identify the connected components after every split. This allows us to construct the single linkage dendrogram efficiently. Instead of storing all the members of each node, it is sufficient to store the range of their Prim's orders (Figure 4).

\begin{figure}[!htb]
	\centering
	\begin{subfigure}[h]{0.49\linewidth}
		\includegraphics[width=0.9\linewidth]{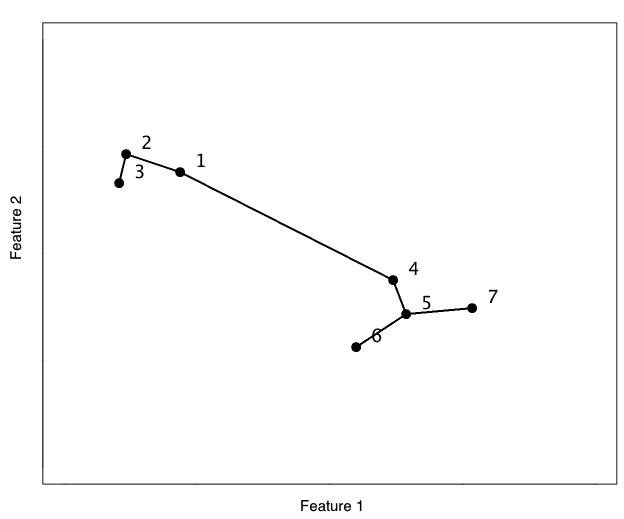}
	\end{subfigure}
	\hfill
	\begin{subfigure}[h]{0.49\linewidth}
		\includegraphics[width=0.9\linewidth]{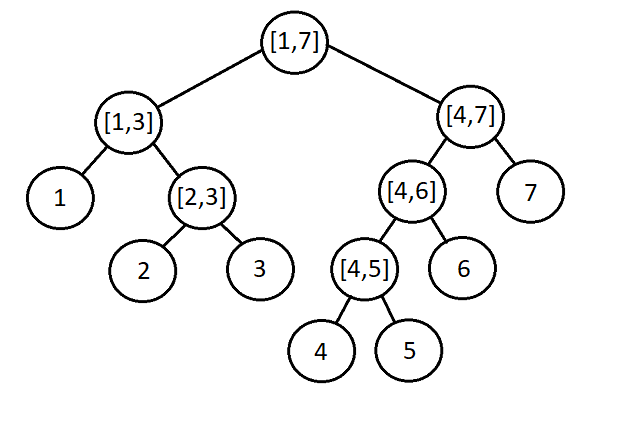}
	\end{subfigure}
	\caption{Single linkage dendrogram with Prim's orders of the node members.}
\end{figure}

\section{The Case of Tied Edge Weights}
Previously we have assumed that there are no tied edge weights in the
graph $G$ and therefore the MST $T$ and Prim’s orders of $G$ (and $T$) are unique. We will now remove this restriction and show that Proposition 1 still holds, with one alteration.

\begin{figure}[!htb]
	\centering
	\includegraphics[width=0.5\linewidth]{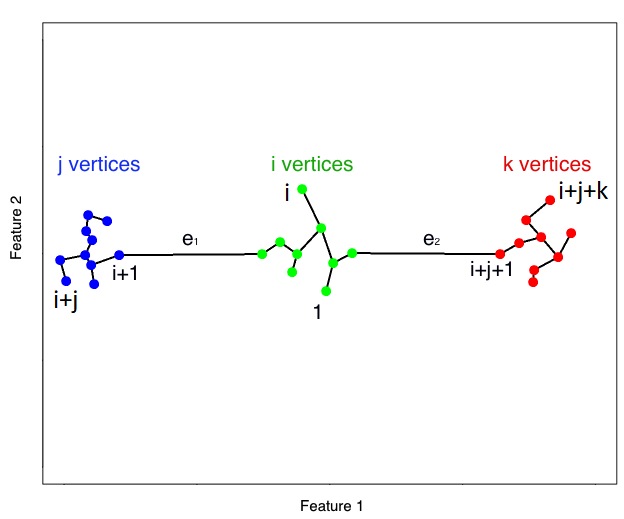}
	\caption{A Prim's Order with Two Longest Tied Edge Weights}
\end{figure}

Figure 5 illustrates the problem. Suppose $T$ has 2 longest edges $e_1$ and $e_2$. Breaking both of them would divide $T$ into three subtrees $T_{green}$, $T_{blue}$, and $T_{red}$. Let $i, j, k$ be the number of vertices in the three subtrees. The fact that there are two longest edges causes ambiguities in tree growing and tree cutting.

Suppose we started Prim's algorithm from a seed vertex in $T_{green}$. Eventually we would have to decide whether to add $e_1$ or $e_2$ next, which would result in different Prim's orders. Let's assume we picked $e_1$ first. Then Prim's order would be as in Figure 5.

Now consider the process of tree cutting. Since edges $e_1$ and $e_2$ have the same length, we need to decide which one to break first. If we broke $e_1$ first, the two connected components would be $T_{blue}$ and $T_{red}+T_{green}$. If Proposition 1 was true, one of the components should have vertices with Prim's order less than $P(e_1) = i+1$, and the other one should have vertices with order greater than or equal to $P(e_1)$. However, this is not the case since the Prim's orders of $T_{blue}$ are $\{(i+1),...,(i+j)\}$ and the Prim's orders of $T_{red} + T_{green}$ are $\{1,...,i\} \cup
\{(i+j+1),...,(i+j+k)\}$. If, on the other hand, we broke $e_2$ first, then the two connected components would be $T_{red}$ and $T_{green} + T_{blue}$, and the corresponding Prim's orders of the vertices would be $\{(i+j+1),...,(i+j+k)\}$ and $\{1,...,(i+j)\}$. Proposition 1 holds in this case. Notice that $P(e_2)>P(e_1)$. This suggests the following proposition:

\begin{proposition}[Generalized version of Proposition 1]
Let $G$ be a connected edge weighted graph. Applying Prim's algorithm to $G$ will result in an MST $T$ and a Prim's order $P$. Let $E_{max}$ be the set of edges with tied longest edge weight. Break the edge $e_{max}=(v_i, v_j) \in E_{max}$ with the largest Prim's order, thereby creating subtrees $T_L$ and $T_R$ with vertex sets $V_L,V_R$ and edge sets $E_L,E_R$. Then 
	\[
	V_L = \{v: P(v) < P(e_{max})\} \text{ and } V_R = \{v: P(v) \geq P(e_{max})\}
	\]
	\[
	E_L = \{e: P(e) < P(e_{max})\} \text{ and } E_R = \{e: P(e) > P(e_{max})\}
	\]
where $P(e_{max})=max(P(v_i),P(v_j))$.
\end{proposition}

We first prove a generalized version of Lemma 1.

\begin{lemma}[Generalized version of Lemma 1]
Let $G$ be a connected edge weighted graph and $T$ be a minimal spanning tree of $G$. If Prim's algorithm is applied to $G$ and $T$ with the same seed vertex, then every Prim's order of $G$ is a Prim's order of $T$ and every Prim's order of $T$ is a Prim's order of $G$.
\end{lemma}

\begin{proof}
Let $E_G$ be the set of edges in $G$ and $E_T$ be the set of edges in $T$. The proof has two directions. \\
$(\Rightarrow)$ Applying Prim's algorithm to $G$ will result in an MST $T$ and a Prim's order $P$ of $G$. Let $e$ be an edge in $E_G-E_T$. Removing $e$ from $G$, $T$ will still be a MST of $G-e$ and $P$ will still be a Prim's order of $G-e$. Repeatedly remove edges from $G$ until $G=T$. Then $P$ is also a Prim's order of $T$. \\
$(\Leftarrow)$ Applying Prim's algorithm to $T$ will define a Prim's order $P_T$ of $T$. Let $e$ be an edge in $E_G-E_T$. Adding $e$ to $T$, $T$ will still be a MST of $T+e$ and $P_T$ will still be a Prim's order of $T+e$. Repeatedly add edges to $T$ until $T=G$. Then $P_T$ is also a Prim's order of $G$.
\end{proof}

\begin{proof}[Proof of Proposition 2]
Based on Lemma 2, we know any Prim's order of $T$ is a Prim's order of $G$ and any Prim's order of $G$ is a Prim's order of $T$. Therefore, it suffices to prove Proposition 2 for $T$ rather than the original graph $G$. Without loss of generality, the seed vertex is in $T_L$. We claim that for any $e\in E_{max}$ such that $e\neq e_{max}$, $e$ must be in $T_L$. Suppose $e\in T_R$, then $e$ must have a larger Prim's order than $e_{max}$ since we must come through $e_{max}$ before joining any edges in $T_R$. However, by hypothesis $P(e)<P(e_{max})$. So this is a contradiction and $e$ must be in $T_L$. For other edges in $T_L$, since $e_{max}$ is a longest edge, this implies that $e_{max}$ won't be chosen until all the edges in $T_L$ have been chosen. Therefore, the same conclusion is drawn as in Proposition 1. 
\end{proof}

\section{Summary}

We address the problem of computing single linkage dendrogram. A possible approach is to: 

\begin{itemize}
	\item Form an edge weighted graph $G$ over the data, with edge weights reflecting dissimilarities.
	\item Calculate the MST $T$ of $G$.
	\item Break the longest edge of $T$ thereby splitting it into subtrees $T_L$, $T_R$.
	\item Apply the splitting process recursively to the subtrees.
\end{itemize}

This approach has the attractive feature that Prim's algorithm for MST construction calculates distances as needed, and hence there is no need to ever store the inter-point distance matrix.

The recursive partitioning algorithm allows us to determine the vertices (and edges) of $T_L$ and $T_R$. We have shown how this can be done easily and efficiently using Prim's order generated by Prim's algorithm without any additional computational cost.

\newpage

\bibliographystyle{apacite}
\bibliography{reference}

\end{document}